\documentclass[10pt]{article}
\usepackage{dcolumn}% Align table columns on decimal point
\usepackage{bm}% bold math
\usepackage{hyperref}
\usepackage{verbatim}       % Defines \begin{comment}, \end{comment}

\usepackage{color}
\usepackage{graphicx}

\usepackage{amsthm}
\usepackage{amssymb}

\usepackage{amstext}
\usepackage{psfrag}
\usepackage{amsmath}
\usepackage{amsfonts}
\usepackage{units}
\usepackage{mathrsfs}

\newcommand{\p}{\partial}

\newtheorem{theorem}{Theorem}[section]
\newtheorem{corollary}[theorem]{Corollary}

\newtheorem{proposition}[theorem]{Proposition}
\theoremstyle{definition}
\newtheorem{definition}[theorem]{Definition}
\theoremstyle{remark}
\newtheorem{remark}[theorem]{Remark}

\newtheorem*{claim*}{Claim}  % unnumbered version

\begin{document}

\title{Global hyperbolicity meets order completeness}

\author{E. Minguzzi\footnote{Dipartimento di Matematica, Universit\`a degli Studi di Pisa,  Via
B. Pontecorvo 5,  I-56127 Pisa, Italy. E-mail:
ettore.minguzzi@unipi.it, ORCID:0000-0002-8293-3802}}

%\date{August 2026, chvsgl.tex}
\date{}
\maketitle

\begin{abstract}
\noindent
Recently, an order completeness property has attracted attention in some low regularity spacetime geometry literature, where it was named {\em chronocompleteness}. In this work we prove that, in the framework of standard Lorentzian geometry, this property is equivalent to global hyperbolicity. The equivalence of global hyperbolicity with other more traditional forms of order completeness is also established.
\end{abstract}

\section{Introduction}

Some years ago Martin and Panangaden \cite{martin06,martin11} introduced an abstract order theoretic generalization of globally hyperbolic spacetimes. A key property of those spaces was a kind of order completeness, meaning that non-decreasing sequences admitting an upper bound had a supremum (and a limit in a suitable topology).

Recently, this type of property has been reconsidered  in optimal transport approaches to  low regularity   Lorentzian geometry \cite{beran24,mondino25,gigli25,braun26,ohanyan26}.

In a standard (smooth or sufficiently regular, say $C^{1,1}$, metric $g$) Lorentzian setting  the following definition  was adopted \cite{beran24,braun26,ohanyan26}

\begin{definition}
The spacetime $(M,g)$ is said to be {\em future chronocomplete} if every chronologically increasing sequence $x_n$ in $M$, $x_n \ll x_{n+1}$, which is upper bounded, in the sense that  there is $x^+\in M$  with $x_n \ll x^+$ for every $n \in \mathbb{N}$, converges to some point of $M$.

The definition of {\em past chronocomplete} is given analogously in a time dual fashion. A spacetime is {\em chronocomplete} if it is both past and future chronocomplete.

Analogous notions of {\em causalcompleteness} are introduced, where $\ll$ is replaced by the causal relation $\le$.
\end{definition}

It was soon pointed out that Martin and Panangaden had already proved that globally hyperbolic spacetimes are chronocomplete and causalcomplete \cite{martin06,braun26}. As the result is a bit disguised in the language of domain theory, help in reading it is provided in the Appendix. In any case, when framed in the language of causality theory the proof of the implications is just a  line. Indeed, as global hyperbolicity implies the non-total imprisoning property, a timelike curve $\gamma: [0,+\infty)\to M$, $\gamma(n)=x_n$, connecting the points cannot be future inextendible, as it would be imprisoned in the compact set $J^+(x_0)\cap J^-(x^+)$, thus it has limit $\lim_{t\to \infty} \gamma(t)=\bar x$ which implies $x_n\to \bar x$ (see \cite[Thm.\ A.4 (i)$\Rightarrow$(iv)]{ohanyan26} for another simple argument).

The problem of the converse implications remained open though the general impression, expressed  by Stefan Suhr and others, was that these  properties were likely equivalent to global hyperbolicity. Although the first guess was that the result had to be more or less straightforward, it remained elusive. Recently, Ohanyan and S\'alamo Candal proved it under closure of the causal relation \cite{ohanyan26}.

Unfortunately, the latter condition is quite strong: it can be recalled that in the relevant 4-dimensional spacetime case, if there is at least one point not contained in a closed  timelike curve (i.e.\ the spacetime is non-totally vicious) --- the weakest causality condition --- the closure of the causal relation implies causal simplicity \cite{minguzzi19c}, a property which sits just below global hyperbolicity in the causal ladder of spacetimes \cite{minguzzi18b}.
%In other words, in order to prove global hyperbolicity from  chronocompleteness they impose a condition that is almost as strong as global hyperbolicity itself.

In this work we establish the equivalence between chrono/causal completeness, in the sense of the above definition, and global hyperbolicity, that is, we prove the reverse implication without imposing  any other assumptions.\footnote{We note that the equivalence of chronocompleteness and causalcompleteness was directly proved in \cite[Lemma 2.A]{ohanyan26},  while in this work it will be derived from their equivalence with global hyperbolicity.}

This is important because, as noted in \cite{ohanyan26}, the result mirrors the equivalence established by the Hopf–Rinow theorem for Riemannian manifolds between the Heine-Borel property (open balls are relatively compact) and completeness. A  link between a type of Heine-Borel property and global hyperbolicity was established in \cite{minguzzi25b}; thus, the result of this work recovers again the equivalence between the Heine-Borel property and completeness, but this time in a spacetime framework. Specifically, from the result of \cite{minguzzi25b} and Thms.\  \ref{btsd} and \ref{cnpf} below, we get

\begin{theorem} \label{nzaw}
On a spacetime $(M,g)$ the following properties are equivalent:
\begin{itemize}
\item[($\alpha$)] the topology generated by
the chronological diamonds is Hausdorff and any chronological diamond is relatively
compact in that topology (Heine-Borel-type property),
\item[($\beta$)] future chronocompleteness (Dedekind-type completeness property).
\item[($\gamma$)] $(M,K)$ is a poset and every non-empty upper bounded directed set has a supremum (Dedekind-type completeness property).
\end{itemize}
\end{theorem}
Here $K$ is the smallest closed, reflexive and transitive relation containing the causal relation.

Of course they are also equivalent to global hyperbolicity, and in ($\beta$) `future' can be replaced with `past' or dropped. Moreover, in ($\beta$) {\em chrono} can be replaced with {\em causal}, but since   ($\alpha$) is expressed in terms of the chronological relation, the chronological version is more natural.

Actually, as shown in  \cite{burtscher22,minguzzi25b},  the same result holds  replacing  future chronocompleteness with the Cauchy completeness with respect to the null distance (or Cauchy completeness with respect to a suitably Wick rotated metric \cite[Cor.\ 4.4]{burgos26}). Unfortunately, the latter  requires the additional ingredient of a time function for its definition  which is a bit unnatural in spacetimes that are not cosmological in the sense of \cite{andersson98}.

%Another form of the equivalence, making use of null distances was given before in \cite{burtscher22}, but that result, based on the notion of null distance, required the additional ingredient of a time function (a bit unnatural in spacetimes that are not cosmological in the sense of \cite{andersson98}).

% well-known equivalence
%between completeness and properness for a Riemannian manifold as a consequence of the
%Hopf–Rinow theorem.
%In this work $(M,g)$ is a smooth Lorentzian spacetime.

Our equivalence proof passes through notions related to the causal boundary construction of Geroch, Kronheimer and Penrose \cite{geroch72}. In fact, we fairly generalize the equivalence between global hyperbolicity and the absence of timelike boundary points (i.e.\ global hyperbolicity is equivalent to the future (or past) causal boundary being achronal) by not demanding causality condition in the reverse implications. It can be noted that GKP impose the distinction property already to establish the equivalence between TIPs and past sets of the form $I^-(\gamma)$ so that the causal boundary theory relies on additional causality conditions which vary depending on the approach \cite{szabados88,marolf03,harris00,flores06b,flores11}. Our result suggests that perhaps there is further room for generalizations and improvements of these causal boundary studies.

\begin{remark}
All the theorems of this work remain valid replacing the Lorentzian spacetime $(M,g)$ with a Finsler spacetime as the relevant causality theory passes through \cite{aazami14,minguzzi15,minguzzi17}. We did not investigate whether they could hold for some versions of Lorentzian metric spaces.
\end{remark}

Before we continue  let us recall our notations and terminologies.
In this work a manifold is always Hausdorff and second countable, hence paracompact. A spacetime $(M,g)$ is a connected time-oriented Lorentzian manifold whose metric $g$ is $C^2$ ($C^{1,1}$ will be enough) and of signature $(-,+,\ldots, +)$.
The inclusion $\subset$ is reflexive, $X\subset X$. With a curve $\gamma$ we might mean a map  $\gamma\colon I \to M$ or the image of the map.
We write $p<q$ if there is a future directed causal curve connecting $p$ to $q$, and $p\ll q$ if there is a future directed timelike curve connecting $p$ to $q$. We write $p\le q$ if $p<q$ or $p=q$. The sets $J=\{(p,q)\colon p\le q\}$ and $I=\{(p,q)\colon p\ll q\}$ are the causal and chronological relations, respectively.
A {\em causal diamond} is a set of the form $J^+(p)\cap J^-(q)$ also denoted $J(p,q)$ (and similarly for {\em chronological diamonds}).
For most results of causality theory we refer to the  review \cite{minguzzi18b} and to the books \cite{penrose72,hawking73,oneill83,beem96}.

\section{Absence of timelike boundary points and equivalence}
%
%\begin{lemma}
%If $p\in M$ does not belong to the chronology violating set, then $J^+(\p I^+(p))\cap D^+(\p I^+(p))=J^+(p)\cap D^+(\p I^+(p))$
%\end{lemma}
%
%\begin{proof}
%Since $p\in \p I^+(p)$  we need only to prove the inclusion $\subset$. Let $q\in J^+(\p I^+(p))\cap D^+(\p I^+(p))$ then there is a causal curve $\sigma$ connecting some point $r\in \p I^+(p)$ to $q$. If $r\in E^+(p)$ then $q\in J^+(p)$ and we have finished. If $r\in  \p I^+(p)\backslash E^+(p)$, by the limit curve theorem, there is a past inextendible timelike curve $\gamma$ contained in $\p I^+(p)$ with future endpoint $r$. The idea now is to push $\gamma$ slightly to the
%
%
%\end{proof}

We shall need the following result  first proved by Penrose under strong causality for $(M,g)$  \cite{penrose79} (see also, (iii) $\Rightarrow$ (i) in
\cite[Thm.\ 3.29]{flores11}, and  the related \cite[Thm.\ 6.2]{budic74} under causal continuity). Remarkably, this causality condition can be dropped leaving  a neat alternative characterization of global hyperbolicity.

\begin{theorem} \label{nnty}
For a spacetime $(M,g)$ the following properties are equivalent:
\begin{itemize}
\item[(i)] Global hyperbolicity.
\item[(ii)] There is no pair $(\gamma, p)$, where $\gamma$ is a future-inextendible timelike curve and $p\in M$,  such that $I^-(\gamma)\subset I^-(p)$.
\item[(iii)] There is no pair $(\gamma,p)$, where $\gamma$ is a past-inextendible timelike curve and $p\in M$,  such that $I^+(\gamma)\subset I^+(p)$.
\end{itemize}
\end{theorem}

Note that since $\gamma$ is timelike it belongs to its own timelike past (resp.\ future), thus in the above items we can replace $I^\mp(\gamma)$ with $\gamma$.

\begin{proof}
(i) $\Rightarrow$ (ii). Assume (i) and suppose that (ii) does not hold, then  there is a pair  $(\gamma, p)$, where $\gamma: [0,+\infty)\to M$ is a future-inextendible timelike curve and $p\in M$,  such that $\gamma\subset I^-(p)$. The timelike curve $\gamma$ is totally future imprisoned in the compact set   $K:=J^+(\gamma(0))\cap J^-(p)$, a contradiction with global hyperbolicity (as this property implies non-total imprisonment).

(ii) $\Rightarrow$ (i). Let us assume (ii). We first prove that $(M,g)$ is non-total imprisoning. Suppose not; then there is a future inextendible causal curve $\sigma\colon[0,b)\to M$ contained in a compact set $C$ (we recall that future non-total imprisonment is equivalent to past non-total imprisonment and so to non-total imprisonment \cite{minguzzi18b}). Its set of future accumulation points $\Omega_f(\sigma)$ is non-empty, as every sequence $\sigma(t_n)$, $t_n\to b$, has a subsequence converging in $C$. Moreover, $\Omega_f(\sigma)$ cannot consist of a single point $e$, for otherwise every such sequence would have $e$ as its only accumulation point in the compact set $C$, so that $\lim_{t\to b}\sigma(t)=e$, in contradiction with the future inextendibility of $\sigma$. Let $p,q\in \Omega_f(\sigma)$ be distinct, let $r\in I^+(q)$, and let $U_p$, $U_q$ be open neighborhoods of $p$ and $q$, respectively, such that $\overline{U_p}\cap\overline{U_q}=\emptyset$ and $U_q\subset I^-(r)$ (recall that $I^-(r)$ is an open neighborhood of $q$).

There are parameters $s_n<t_n<s_{n+1}$, $s_n\to b$, with $p_n:=\sigma(s_n)\in U_p$ and $q_n:=\sigma(t_n)\in U_q$ for every $n$; since $\sigma$ is causal, $p_n\le q_n\le p_{n+1}$. Let $x\in I^-(p_1)$ and let $\gamma_1$ be a timelike curve from $x$ to $p_1$. Since $U_p$ is open, $\gamma_1$ contains a point $\tilde p_1\in U_p$, $\tilde p_1\ne p_1$, so that $\tilde p_1\ll p_1\le q_1$ and hence $\tilde p_1\ll q_1$. Let $\sigma_1$ be a timelike curve from $\tilde p_1$ to $q_1$; it contains a point $\tilde q_1\in U_q$, $\tilde q_1\ne q_1$, so that $\tilde q_1\ll q_1\le p_2$ and hence $\tilde q_1\ll p_2$. Iterating, we obtain points
\[
\tilde p_1\ll \tilde q_1\ll \tilde p_2\ll \tilde q_2\ll\cdots, \qquad \tilde p_n\in U_p,\ \tilde q_n\in U_q,
\]
and, concatenating the timelike segments, a future directed timelike curve $\gamma$ through all of them. The curve $\gamma$ is future inextendible: if it had a future endpoint $e$, then $\tilde p_n\to e$ and $\tilde q_n\to e$, so $e\in \overline{U_p}\cap\overline{U_q}=\emptyset$. Moreover, every point of $\gamma$ lies in the chronological past of some $\tilde q_n\in U_q\subset I^-(r)$, thus $\gamma\subset I^-(r)$, which contradicts (ii). Hence $(M,g)$ is non-total imprisoning and, in particular, causal (a closed causal curve, traversed infinitely many times, would give a totally imprisoned future inextendible causal curve).

Global hyperbolicity is equivalent to ``non-total imprisonment and relative compactness of chronological diamonds'' \cite[Cor.\ 3.3]{minguzzi08e}; thus, it remains to prove that the chronological diamonds are relatively compact.
The following argument is inspired by Penrose \cite{penrose79}, but strong causality is not needed.
We shall use both the domain of dependence $D^-(S)$ defined via causal curves and the one $\tilde D^-(S)$ defined via timelike curves (the latter is denoted $D^-(S)$ by Penrose \cite{penrose72}), recalling that $\textrm{Int} D^-(S)=\textrm{Int} \tilde D^-(S)$ \cite[Prop.\ 3.3]{minguzzi18b}.

Let $p,q\in M$. If $p\not\ll q$, then $I^+(p)\cap I^-(q)=\emptyset$, which is relatively compact, so we can assume $p\ll q$. The set $S=\p I^-(q)$ is a closed achronal set. By chronology $q\notin I^-(q)$, thus $q\in S$ and hence $I^-(S)=I^-(q)$. We claim that $I^-(q)\subset \tilde D^-(S)$. Indeed, let $x\in I^-(q)$ and let $\gamma$ be a future inextendible timelike curve issued from $x$ (the portion of a curve through $x$ preceding $x$ lies in $I^-(x)\subset I^-(q)$, hence does not meet $S$, so it is enough to consider curves issued from $x$). If $\gamma$ did not intersect $S=\p I^-(q)$, then, starting in the open set $I^-(q)$, it would remain entirely in $I^-(q)$, in contradiction with (ii). This proves the claim. As $I^-(q)$ is open, we get $I^-(q)\subset \textrm{Int}\, \tilde D^-(S)=\textrm{Int} D^-(S)$, and in particular $p\in \textrm{Int} D^-(S)$. By \cite[Thm.\ 5.20]{penrose72} \cite[Prop.\ 6.6.6]{hawking73} \cite[Thm.\ 3.47, Prop.\ 3.33 Eq.\ (3.10)]{minguzzi18b} the set $J^+(p)\cap J^-(S)$ is compact. It contains $I^+(p)\cap I^-(S)=I^+(p)\cap I^-(q)$, which is therefore relatively compact.

This concludes the proof of the equivalence between (i) and (ii). That between (i) and (iii) is analogous.
\end{proof}
%\begin{definition}
%The spacetime $(M,g)$ is said to be {\em future chronocomplete} if every chronologically increasing sequence $x_n$ in $M$, $x_n \ll x_{n+1}$, which is upper bounded, in the sense that  there is
%$x^+\in M$  with $x_n \ll x^+$ for every $n \in \mathbb{N}$ converges to some points of $M$.

%The definition of {\em past chronocomplete} is given analogously in a time dual fashion. A spacetime is {\em chronocomplete} if it is both past and future chronocomplete.
%\end{definition}

\begin{theorem} \label{btsd}
The following conditions on $(M,g)$ are equivalent:
\begin{itemize}
\item[(a)] Global hyperbolicity.
\item[(b)] Future chronocompleteness.
\item[(c)] Past chronocompleteness.
\item[(d)] Future causalcompleteness.
\item[(e)] Past causalcompleteness.
\end{itemize}
%A spacetime if globally hyperbolic iff it is chronocomplete.
\end{theorem}

As mentioned, the direction  `(a) $\Rightarrow$ (b) and (d)' was already proved by Martin and Panangaden \cite{martin06} (see our Appendix).
This is the simpler direction. The reverse direction `(b) or (d) $\Rightarrow$ (a)' uses Thm.\ \ref{nnty}.

\begin{proof}
We give the proof for the future notions, the proof for the past notions being analogous.

(a) $\Rightarrow$ (b) and (d). This direction is discussed in the Introduction.
%
%Assume global hyperbolicity and let $x_n$ be a chronologically increasing sequence chronologically  upper bounded by $x^+$. Alternatively, let  $x_n$ be a causally increasing sequence causally upper  bounded by $x^+$.
%
%The set $K:=J^+(x_1)\cap J^-(x^+)$ is compact and $x_n\in K$ for every $n$. Thus there is a an accumulation point $\hat x\in K$. Suppose that $x_n$ does not converge to $\hat x$, then we can find a neighborhood $O\ni \hat x$ and a subsequence $x_{n_k}$ such that $x_{n_k}\in K\backslash O$. Since $K\backslash O$ is compact $x_{n_k}$ accumulates on a second point $\check x$ distinct from $\hat x$. This means that the original sequence accumulates on both $\hat x$ and $\check x$, but then we can find a subsequence $y_r:=x_{n(r)}$ such $y_{2n+1}\to \hat x$, $y_{2n} \to \check x$, for $n\to +\infty$. Since $y_{2n} \le y_{2n+1}$ and $y_{2(n-1)+1}\le y_{2n}$, and the causal relation is closed, we obtain in the limit $\check x \le \hat x\le \check x$. The violation of causality gives a contradiction which proves that $x_n\to \check x$. This shows that global hyperbolicity implies future chronocompleteness.

(b) or (d) $\Rightarrow$ (a). The implication (d) $\Rightarrow$ (b) is trivial, so it suffices to prove  (b) $\Rightarrow$ (a). We prove global hyperbolicity in the version of Theorem \ref{nnty}(ii).  By contradiction, suppose there is a future inextendible timelike curve $\gamma: [0, b)\to M$ and a point $x^+$ such that $\gamma \subset I^-(x^+)$. Let $x_n=\gamma(s_n)$, $s_n\nearrow b$, then $x_n\ll x_{n+1}$ and $x_n\ll x^+$ for every $n$, thus by (b) the sequence $x_n$ has limit $\bar x$. It follows that $\lim_{t\to b} \gamma(t)=\bar x$ otherwise there would be a subsequence $y_n=\gamma(t_n)$, $t_n<t_{n+1}$, $t_n\to b$, such that $y_n\in M\backslash O$ where $O$ is an open neighborhood of $\bar x$. But then  $y_n\ll y_{n+1}$ and $y_n\ll x^+$ and again by (b) $y_n$ has some limit $\check x\ne \bar x$. However, we can find a third sequence $z_k$ consisting of points alternatively taken from $x_n$ and $y_n$ in such a way that $z_k\ll z_{k+1}\ll x^+$. This has also to converge by (b), but it cannot as it has subsequences converging to different points. The contradiction proves that $\lim_{t\to b} \gamma(t)=\bar x$, but this again is a contradiction with the future inextendibility of $\gamma$. The contradiction proves that the property of  Theorem \ref{nnty}(ii) holds, which coincides with global hyperbolicity.
\end{proof}

\section{Order theoretic versions using $K$ or directed sets}
Since in the definition of chrono/causal completeness the notion of limit is mentioned, the property makes use of the manifold topology of $M$. As such it is not entirely order theoretic. We can remedy this problem using a different causal relation in $(M,g)$ and replacing limits with suprema.

Let us denote with $K$ the smallest closed, reflexive and transitive relation that contains $\le$ (Sorkin and Woolgar's relation) \cite{sorkin96,dowker00,janardhan08,minguzzi18b}. We write $x\le_K y$ for $(x,y)\in K$.  The spacetime is said to be {\em $K$-causal} if $K$ is antisymmetric; this property is equivalent to stable causality \cite{minguzzi18b}.
\begin{theorem}
The following properties are equivalent:
\begin{itemize}
\item[(i)] global hyperbolicity,
\item[(ii)] future causalcompleteness
\item[(iii)] $\le_K$ is antisymmetric and   every  increasing sequence $x_n$ in $M$, $x_n \le_K x_{n+1}$, which is upper bounded, in the sense that  there is
$x^+\in M$  with $x_n \le_K x^+$ for every $n \in \mathbb{N}$, admits a  supremum $\bar x$ (i.e.\ $\bar x$ is an upper bound and any other upper bound $\tilde x$ satisfies $\bar x \le_K \tilde x$).
\end{itemize}
Moreover, the limit point of the sequence $x_n$ due to (ii) coincides with the supremum in (iii) and  the supremum is unique.
\end{theorem}

The antisymmetry condition in (iii) is necessary, for a totally vicious spacetime would satisfy $\le\, =M\times M$ so $K=M\times M$ and  every point would be a supremum, but neither (i) nor (ii) would hold. Note that the antisymmetry in (iii) implies that the supremum there mentioned is unique.

We also considered (iii) with $\le_K$ replaced by $\le$ but could only prove the non-total imprisoning property starting from (iii).

\begin{proof}
We already know that (i) and (ii) are equivalent.

(i) and (ii) $\Rightarrow$ (iii). By (i) $\le$ is antisymmetric and coincident with $\le_K$ ($J=K$). Let $x_n$ be a sequence as in (iii), then by (ii) $x_n\to \tilde x$ for some $\tilde x \in M$. Since $J$ is closed and $(x_k,x_n)\to (x_k, \tilde x)$ for $n\to +\infty$, we get $x_k\le \tilde x$, thus $\tilde x$ is an upper bound. If $\check x$ is another upper bound, $x_k\le \check x$, taking the limit $k\to +\infty$ and using again the closure of $J$, $\tilde x\le \check x$. Thus $\tilde x$ is a supremum and we can set $\bar x:=\tilde x$.

(iii) $\Rightarrow$ (i).
Since $K$ is antisymmetric the spacetime is stably causal \cite{minguzzi18b}.

We prove global hyperbolicity in the version of Theorem \ref{nnty}(ii).  By contradiction, suppose there is a future inextendible timelike curve $\gamma: [0, b)\to M$ and a point $x^+$ such that $\gamma \subset I^-(x^+)$. Let $x_n=\gamma(s_n)$, $s_n\nearrow b$, then $x_n\ll x_{n+1}$ and $x_n\ll x^+$ for every $n$, thus by (iii) the sequence $x_n$ has  $K$-supremum $\tilde x$. It is not possible that for every choice of $s_n$, $x_n$ converges to some point for otherwise, arguing as in the proof of Theorem \ref{btsd}, $\gamma(t)$ itself has future limit which contradicts its future inextendibility. This means that we can choose $s_n$ so that $x_n$ does not converge. In particular, we can assume, passing to a subsequence if necessary, that there is a convex relatively compact open neighborhood $O\ni \tilde x$ such that $x_n\notin \bar O$. By \cite[Lemma 3]{minguzzi09c} there is $z_n\in \p O$ such that $x_n\le_K z_n \le \tilde x$ (this result is clear if there is a causal connecting curve between $x_n$ and $\tilde x$, the fact that it holds in the version mentioned is non-trivial as $K$ is not characterized via curves). Let us pass to a further subsequence denoted in the same way so that we can assume $z_n\to z\in \p O$ (note that passing to a subsequence does not change the upper bounds and suprema), so that $z\le_K \tilde x$. Let $w\gg z$, $w\ne \tilde x$, then, for sufficiently large $n$,  $x_n\le_K z_n\ll w$, which means that $w$ is an upper bound and so $\tilde x\le_K w$. Taking the limit $w\to z$ we get $\tilde x\le_K z$ by the closure of $K$. This gives a contradiction with the antisymmetry of $K$ and concludes the proof.
%Let $x_n$ be a sequence as in (ii).
%There cannot be two distinct suprema $y,z$ for a sequence as in (iii), for it would be $y\le_K z\le_K y$ which contradicts antisymmetry of $\le_K$. Thus the supremum in (iii) is unique. Let us denote it $\tilde x$. If $x_n$ converges to $\tilde x$
\end{proof}

The theorem can be restated as follows ({\em poset} is a shorthand for partially ordered set)

\begin{theorem} \label{cngp}
The globally hyperbolic property for $(M,g)$ is equivalent to the following conditions on $(M,K)$:
\begin{itemize}
\item[(i)] it is a poset,
\item[(ii)] every upper bounded increasing sequence has a supremum.
\end{itemize}
\end{theorem}

%\subsection{Directed sets}

An even nicer version would use directed sets rather than increasing sequences.  A subset \(D \subset M\) is \emph{directed} if for every \(d_1, d_2 \in D\), there exists \(d_3 \in D\) such that \(d_1 \leq d_3\) and \(d_2 \leq d_3\).

A {\em closed ordered space} is a topological space endowed with a reflexive, transitive, antisymmetric relation (order or partial order depending on the terminology) which is closed in the product topology, e.g.\ a globally hyperbolic spacetime where the topology is the manifold topology and the relation is the causal relation. We can abstract the notion of global hyperbolicity to the framework of closed ordered spaces by calling them {\em globally hyperbolic closed ordered spaces} if the sets (diamonds) $\{x: a\le x\le b \}$ are compact \cite{minguzzi23}. So a globally hyperbolic spacetime is a globally hyperbolic closed ordered space (this is also true in the Lorentz-Finsler case).

First, we need an important theorem

\begin{theorem}[Bounded Directed Completeness]
\label{thm:main}
%Let $(M,g)$ be a globally hyperbolic spacetime and consider the closed order provided by the causal relation. ore generally,
Consider a globally hyperbolic closed ordered space and denote with $J$ or $\le$ its relation.

Let \(D \subset M\) be a non-empty directed set with an upper bound. Then \(D\) has a supremum.
\end{theorem}

%In particular, the set $U$ of upper bounds of $D$ is downward directed: given $u_1,u_2\in U$, the supremum $w=\sup D$ belongs to $U$ and satisfies $w\le u_1$, $w\le u_2$.

\begin{proof}
Let $U\neq\emptyset$ be the set of upper bounds of $D$ and fix $d_0\in D$,
$u\in U$. Let $J(d_0,u):=J^+(d_0)\cap J^-(u)$ and consider the family of sets
\[
\mathcal{F}=\{J^+(a)\colon a\in D\}\cup\{J^-(z)\colon z\in U\}\cup\{J(d_0,u)\},
\]
which are closed because the relation is closed, the last one being compact by
global hyperbolicity. The family has the finite intersection property: given
$a_1,\dots,a_n\in D$ and $z_1,\dots,z_m\in U$, by directedness there is
$d\in D$ with $d_0\le d$ and $a_i\le d$ for all $i$, and since $u$ and the
$z_j$ are upper bounds, $d\le u$ and $d\le z_j$ for all $j$. Thus
$d\in J(d_0,u)\cap\bigcap_i J^+(a_i)\cap\bigcap_j J^-(z_j)$.
As all members of $\mathcal F$ are closed and one of them is compact, there is
$t\in\bigcap\mathcal F$. Then $a\le t$ for every $a\in D$, so $t\in U$, and
$t\le z$ for every $z\in U$. Hence $t$ is the least upper bound of $D$; it is
unique by antisymmetry.
\end{proof}

\begin{corollary}[Convergence to the supremum]
\label{cor:convergence}
Under the hypotheses of Theorem \ref{thm:main}, regard $D$ as a net
$(d)_{d\in D}$ indexed by itself, the direction being given by the relation
$\le$ restricted to $D$. Then this net converges to $\sup D$ in the topology of the globally hyperbolic closed ordered space.
\end{corollary}

\begin{proof}
Let $t=\sup D$ (Theorem \ref{thm:main}), let $u$ be an upper bound of $D$ and
fix $d_0\in D$. For every $a\in D$ the net is eventually in the set $J^+(a)$
(namely for $d\ge a$). For every upper bound $z$ of $D$ the net lies entirely
in $J^-(z)$. These sets are closed because the relation is closed. In
particular, the net is eventually in the compact set $J(d_0,u)$.

Every cluster point $y$ of the net coincides with $t$. Indeed, since the net is eventually in the closed set $J^+(a)$, we have $a\le y$ for every $a\in D$, so $y$ is an upper bound of $D$. Since the net lies in the closed set $J^-(z)$, we have $y\le z$ for every upper bound $z$. Thus $y$ is the least upper bound, $y=t$.

Suppose the net does not converge to $t$. Then there is an open neighborhood
$O\ni t$ such that $E:=\{d\in D\colon d\notin O\}$ is cofinal in $D$. Thus $E$ is directed and $(e)_{e\in E}$ is a subnet, eventually contained in the
compact set $J(d_0,u)\backslash O$. Hence it has a cluster point
$y\notin O$, which is also a cluster point of the original net. By the
previous paragraph $y=t\in O$, a contradiction.
\end{proof}

\begin{theorem} \label{cnpf}
The globally hyperbolic property for $(M,g)$ is equivalent to the following conditions on $(M,K)$:
\begin{itemize}
\item[(i)] it is a poset,
\item[(ii')] every non-empty upper bounded directed set has a supremum.
\end{itemize}
\end{theorem}

In poset theory (ii') is called {\em bounded directed completeness} while (ii) is called {\em  bounded
$\omega$-chain completeness}. The former is stronger than the latter which, however, is sufficient in analysis, measure theory and probability.

\begin{proof}
Since every increasing sequence is directed, it is clear that (ii') implies Thm.\ \ref{cngp}(ii) and so global hyperbolicity.

For the converse, let us assume global hyperbolicity. It implies $K$-causality (equiv.\ stable causality) and so (i). Moreover, $\le_K=\le$ thus (ii') is precisely the assumption of Theorem \ref{thm:main}, which implies that the directed set has a supremum.
%Let $D$ be a directed set which is bounded by $x$. For each $y\in D$, consider the compact set $C_y:=J^+(y)\cap J^-(x)$ (which is non-empty as it contains $x$). The non-empty closed set $C=\cap_{y\in D} C_y$ (it contains $x$) is compact as it is contained in $C_z$, for any $x\in D$. Any point of $C$ is clearly an upper bound.
\end{proof}

Let $D_p=\{(x,y): I^-(x)\subset I^-(y)\}$,  $D_f=\{(x,y): I^+(x)\supset I^+(y)\}$ and $D=D_p\cap D_f$. They are all reflexive and transitive relations. The spacetime is said to be {\em weakly distinguishing} if $D$ is antisymmetric.  It is {\em reflecting} if $D_f=D_p$ (equiv.\ $D=\bar J$, see \cite[Def.\ 4.9]{minguzzi18b}).
\begin{theorem} \label{cntf}
Global hyperbolicity for $(M,g)$ is equivalent to the following three conditions, the last two of which are about $(M,D)$
\begin{itemize}
\item[(a)] $(M,g)$ is reflecting,
\item[(b)] $(M,D)$ is a poset,
\item[(c)] on $(M,D)$ every non-empty upper bounded directed set has a supremum.
\end{itemize}
\end{theorem}

An alternative nice form, which involves only $D_p$, is as follows. Replace $D$ with $D_p$ in (b) and (c) and write (a) as: $D_p$ coincides with the same relation for the opposite time orientation. Since that relation is $D_f$, condition (a) is  $D_p=D_f$, namely reflectivity, and so $D=D_p$, showing the equivalence of the two formulations.

\begin{proof}
Suppose $(M,g)$ is globally hyperbolic, then $D=J=K$, causality holds and so (b) holds, reflectivity holds and (c) follows from the validity of  (ii') of Thm.\ \ref{cnpf} using $K=D$.

Assume (a),(b) and (c).
By reflectivity $D=\bar J$ thus $D$ is  closed. Clearly $D$ is the smallest closed reflexive and transitive relation containing $J$, thus $D=K$. As $(M,D)$ is a poset, $(M,K)$ is a poset. Condition (c) using $D=K$ reads as (ii') of Thm.\ \ref{cnpf}, thus from that theorem $(M,g)$ is globally hyperbolic.
\end{proof}

\section{Conclusions}
We have shown that the causal completeness notions recently used in the
low-regularity and optimal transport approaches to Lorentzian geometry ---
future/past chronocompleteness and future/past causalcompleteness --- are
all equivalent  to global hyperbolicity, with no auxiliary
hypothesis (Thm.\ \ref{btsd}). This confirms the expectation among experts in the field that these Dedekind-type completeness properties are not new
causality conditions but simply global hyperbolicity itself, seen from an
order-theoretic vantage point.

The key tool is Thm.\ \ref{nnty}, a characterization of global hyperbolicity
as the absence of a future- (or past-) inextendible timelike curve
chronologically trapped in the past (future) of a single point --- a condition formally equivalent to
the absence of timelike points in the causal boundary of Geroch, Kronheimer
and Penrose.
%%Unlike earlier versions of this characterization, due to Penrose
%\cite{penrose79} and, under causal continuity, to Budic and Sachs
%\cite{budic74}, no a priori causality condition needs to be imposed.

Combined with the Heine-Borel-type characterization of \cite{minguzzi25b},
our results place global hyperbolicity in a genuine Hopf--Rinow-type
correspondence (Thm.\ \ref{nzaw}): properness of the causal topology on one
side, Dedekind completeness of the causal order on the other,
paralleling properness versus metric completeness for Riemannian manifolds.

Moreover, the equivalence can be stated in purely order‑theoretic language, dispensing with the manifold topology altogether. Replacing the causal relation with the Sorkin–Woolgar relation $K$, global hyperbolicity is equivalent to $(M,K)$  being a poset in which every upper-bounded increasing sequence has a supremum (Thm. \ref{cngp}), and, more fundamentally, to the condition that every non-empty upper-bounded directed set has a supremum (Thm. \ref{cnpf}). The latter is precisely the bounded directed completeness of poset theory, a property that underlies much of analysis and measure theory.
%This reformulation shows that the Dedekind-type completeness is intrinsic to the causal order itself, independent of any auxiliary topological or metric structure, and places the Hopf-Rinow-type correspondence  in the realm of ordered spaces.

\section*{Appendix: original Domain Theory approach}

In this section we show how to read the implication ``global hyperbolicity $\Rightarrow$ future chronocompleteness'' in the Domain Theory approach by Martin and Panangaden \cite{martin06}. This should not be immediately obvious to readers not familiar with their work. For more on this theory see \cite{keimel09b,ebrahimi14,ebrahimi15,finster21,sharifzadeh19,mazibuko23}.

The statement is essentially  \cite[Thm.\ 6.1(iv)]{martin06}
\begin{proposition} \label{cnyt}
On a globally hyperbolic poset each  non-empty directed set with an upper bound has a supremum.
\end{proposition}

We need to clarify the terms entering this result. A poset  is a set endowed with a reflexive, transitive antisymmetric relation. They define a {\em globally hyperbolic poset} out of two abstract relations \cite[Sec.\ 6]{martin06}

\begin{definition}
A {\em globally hyperbolic poset} is a poset $(X,\le)$ such that
\begin{itemize}
\item[(a)] $(X,\le)$ is bicontinuous, and
\item[(b)] intervals $[a,b]$  are compact in the interval topology — a topology derived purely from the order \cite[Def. 2.12]{martin06}, with no reference to any manifold or pre-existing topology.
\end{itemize}
\end{definition}
They prove \cite[Sec.\ 6]{martin06} that globally hyperbolic spacetimes $(M,g)$ provide examples of globally hyperbolic posets where $\le$ is identified with the causal relation.\footnote{Theorem \ref{thm:main} and Prop.\ \ref{cnyt} have different hypotheses: the former needs an a priori topology in which the order is closed and the diamonds are compact, but no bicontinuity; the latter needs bicontinuity, with compactness understood in the interval topology derived from the order. Both apply to globally hyperbolic spacetimes, where they show that every non-empty directed set with an upper bound has a supremum.}
%\footnote{Our Theorem \ref{thm:main} is different in scope than Prop.\ \ref{cnyt} in that our topology is arbitrary and only such that the causal diamonds are compact, moreover the poset is not required to be bicontinuous. Both results applied to globally hyperbolic spacetimes $(M,g)$ are sufficient to deduce that every non-empty directed set with an upper bound has a supremum.}
In this case the interval topology coincides with the Alexandrov topology and $\ll$ --- a relation deduced entirely in an order theoretic way from $(X,\le)$ (see \cite[Def.\ 2.5]{martin06})  (the condition of being bicontinuous provides a compatibility between alternative definitions of $\ll$) --- coincides with the chronological relation $I$.

Now, consider a globally hyperbolic spacetime $(M,g)$. If we have a chronologically increasing sequence $x_n$,  $x_n\ll x_{n+1}$, which is chronologically upper bounded $x_n\ll x^+$, then the same is true causally $x_n\le x_{n+1}$, $x_n\le x^+$, and so there is a supremum $\bar x$. Let us consider sequences $r_k\to \bar x$, $q_k\to \bar x$, $r_{k}\ll r_{k+1} \ll \bar x \ll q_{k+1} \ll q_k$, then convergence $x_n\to \bar x$ in the manifold topology (which under global hyperbolicity and hence strong causality is equivalent to the Alexandrov topology) is equivalent to: for every $k$ we have for sufficiently large $n$, $r_k\ll x_n \ll q_k$. But $x_n\ll x_{n+1} \le \bar x\ll q_k$ because $\bar x$ being a supremum is an upper bound, thus $x_n \ll q_k$ holds trivially for every $n$. As for $r_k\ll x_n$ this is precisely consequence of the abstract poset definition of $\ll$ (which in the above expressions we use interchangeably with $I$ as they coincide) \cite[Def.\ 2.5]{martin06} as $r_{k+1}\ll \bar x$ implies, as $\bar x$ is a supremum, $r_k\ll r_{k+1}\le x_n$ for some $n$ and hence for every larger $n$.

This shows how Prop.\ \ref{cnyt} implies the spacetime manifold statement ``global hyperbolicity $\Rightarrow$ future chronocompleteness''. The proof of ``global hyperbolicity $\Rightarrow$ future causalcompleteness'' is almost the same, in fact it is simpler.

\section*{\normalsize Data availability statement}
No new data were created or analysed in this study.

\section*{\normalsize Conflict of Interest statement}
The author of this publication declares no conflict of interest.

%\bibliography{../../bibliografie/simultaneity,../../bibliografie/libri,../../bibliografie/miei,../../bibliografie/mieiPrep,../../bibliografie/mieiProc}
%\bibliographystyle{plain}

\end{document}